\documentclass{article}
\usepackage{spconf,amsmath,graphicx,hyperref}

\usepackage[utf8]{inputenc} 
\usepackage[T1]{fontenc}    
\usepackage{hyperref}       
\usepackage{url}            
\usepackage{booktabs}       
\usepackage{amsfonts}       
\usepackage{nicefrac}       
\usepackage{microtype}      
\usepackage{xcolor}         
\usepackage{lipsum}
\usepackage{amsmath, amssymb, amsthm}
\usepackage{thmtools}
\usepackage{thm-restate}
\usepackage{graphicx} 
\usepackage{cite}
\usepackage{enumitem}
\usepackage{adjustbox}
\usepackage[utf8]{inputenc} 
\usepackage[T1]{fontenc}    
\usepackage{hyperref}       
\usepackage{url}            
\usepackage{booktabs}       
\usepackage{amsfonts}       
\usepackage{nicefrac}       
\usepackage{microtype}      
\usepackage{enumitem}       
\usepackage{amsmath}        
\usepackage{amsthm}         
\usepackage{bbm}            
\usepackage{xcolor}         
\usepackage{soul}           
\usepackage[ruled,vlined]{algorithm2e}    
\usepackage{graphicx}       
\usepackage{subcaption}      
\usepackage{booktabs}       
\usepackage{multirow}       
\usepackage{wrapfig}        

\usepackage[ruled,vlined]{algorithm2e}    

\let\OLDthebibliography\thebibliography
\renewcommand\thebibliography[1]{
  \OLDthebibliography{#1}
  \setlength{\parskip}{1pt}
  \setlength{\itemsep}{1pt plus 0.3ex}
}

\allowdisplaybreaks

\allowdisplaybreaks

\newtheorem{definition}{Definition}[section]
\newtheorem{proposition}{Proposition}
\newtheorem{assumption}{Assumption}
\newtheorem{lemma}{Lemma}

\title{Identifying common backbones of interactions underlying \\ food webs via non-deterministic alignments}
\name{Yifan Xu$^{*}$, Carlos Taveras$^{*}$, Lydia Beaudrot$^{\dagger}$, C\'esar A. Uribe$^{*}$}
\address{$^{*}$Rice University, Houston, TX 77005, USA\\
$^{\dagger}$Michigan State University, East Lansing, MI 48824, USA}
\begin{document}
%
\maketitle
\begin{abstract}
Climate change reshapes food webs by altering species distributions and interactions, making it essential to identify structural backbones of interactions that persist across ecosystems. Deterministic alignment methods are computationally slow and restricted to one-to-one correspondences. We introduce a scalable, non-deterministic alignment framework inspired by optimal transport that captures overlapping species roles via many-to-many mappings. Framed via motif-role profiles as a Gromov–Wasserstein transport problem, our method is both efficient and interpretable. We apply the proposed method to a large continental-scale data set of 129 mammal food webs in Sub-Saharan Africa. Pairwise alignments are identified, and we uncover robust backbones with greater connectivity and transitivity than those expected under null models. The proposed approach provides a formal, reproducible tool for forecasting ecosystem reorganization and conservation efforts.
\end{abstract}
\begin{keywords}
Graph alignment, ecological networks, optimal transport
\end{keywords}
\section{Introduction}
\label{sec:intro}

Food webs are directed networks that model predator-prey interactions and map the flow of energy through ecosystems. \textit{Climate change alters these networks} by reshaping species distributions, modifying interaction strengths, and, in some cases, dismantling trophic pathways altogether~\cite{lurgi2012climate}. Such perturbations hinder predictions of ecosystem responses, as changes in interaction structure can cascade to affect biodiversity and ecosystem services \cite{oliver2015biodiversity}. This makes food web analysis indispensable for understanding the stability, functioning, and resilience of ecological communities under environmental stressors~\cite{oliver2015biodiversity,thebault2010stability,barnes2018energy}.  In this context, identifying structural patterns that persist across food webs, referred to as the \textbf{backbone of interactions}, becomes crucial~\cite{bramon2018identifying}. These \textit{backbones} highlight recurrent trophic links that endure despite climate-driven variability, revealing ecological regularities that constrain community assembly. By isolating these persistent features, food web analysis not only deepens our theoretical understanding of ecosystem organization but also provides a framework for forecasting how ecosystems might reorganize under future climate scenarios.

\textit{Backbones of interactions} are typically identified through comparative analyses of ecological networks across communities, ecosystems, or large datasets. Traditionally, such comparisons have focused on network-level descriptors, e.g., connectance, modularity, and trophic-level distributions, that capture broad structural features of food webs \cite{dunne2002network, williams2000simple, stouffer2011role}. While informative, these descriptors provide only coarse-grained descriptions and often fail to capture species-level similarities or functional roles. To overcome this limitation, recent approaches emphasize direct species alignments based on their structural positions within networks. A pioneering effort by \cite{bramon2018identifying} introduced a deterministic network alignment framework that maps species one-to-one across food webs, identifying analogous roles through motif-role similarity \cite{milo2002network, baiser2016motifs, cenci2018rethinking, paulau2015motif, tavella2022using}. Despite its fine-grained perspective, this deterministic strategy has drawbacks that hinder its wide-scale applicability. Specifically, it scales poorly as food web size increases and enforces the restrictive assumption of strict one-to-one role correspondence, potentially missing cases where ecological roles overlap or only partially align~\cite{stouffer2012evolutionary}. Moreover, the reliance on annealing-based optimization adds technical complexity, making alignments difficult to reproduce and computationally expensive for statistically significant results, limiting applicability to large-scale food webs.

We introduce a \textbf{non-deterministic} alignment framework, coupled with an optimization-based method, to identify the backbones of interactions in food webs. Analogous to the relaxation introduced by Kantorovich to the classic Monge's Optimal Transport (OT) problem \cite{ villani2008optimal,kantorovich1942translocation}, we relax the deterministic pairings into many-to-many alignments. In summary, our contributions are as follows:
\begin{enumerate}[ leftmargin=2.0em,left=-1pt,itemsep=1pt, parsep=-1pt, topsep=-0pt, partopsep=-1pt]
    \item We provide a first mathematically rigorous formulation of ecological network alignment based on network motifs. 
    \item We propose a computationally efficient, scalable, non-deterministic alignment framework and algorithm.
    \item We introduce the notion of backbone of interactions for non-deterministic alignments and show they satisfy the criteria provided by \cite{bramon2018identifying}.
\end{enumerate}

\section{Computing Non-deterministic Motif-based  Network Alignments}

Aligning two ecological networks amounts to identifying correspondences between their species based on the ecological roles they play within their respective networks. Formally, our task is to find an \textit{alignment} (defined below) that minimizes a cost functional quantifying the aggregated dissimilarity between species across two food webs.

\vspace{-0.2cm}
\begin{definition}\label{def:non-det alignment}
Let $G_1 = (V_1,E_1)$ and $G_2=(V_2,E_2)$ denote two directed networks where $m=|V_1|$ and $n=|V_2|$. Moreover, let ${\mu} \in \Delta^{m - 1}, {\nu} \in \Delta^{n - 1}$ be two distributions representing species importance. A \textbf{non-deterministic alignment} between $G_1,G_2$ is a matrix $T \in \mathcal{C}_1({\mu}) \cap \mathcal{C}_2({\nu})$, where

\vspace{-0.7cm}
\begin{align} \label{eq:marginal_constraints}
\mathcal{C}_1(\mu) &:= \{T \in [0, 1]^{m \times n}  \; | \; T\mathbbm{1}_{n} \preceq \mu\}, \quad \text{and} \nonumber \\ \quad \mathcal{C}_2(\nu)& := \{T \in [0, 1]^{m \times n} \; | \; T^\top \mathbbm{1}_m \preceq \nu\},
\end{align}
\vspace{-0.7cm}

\noindent with $T_{ij}$ representing the alignment between species $i$ and $j$.
\end{definition}

\vspace{-0.2cm}
Next, we build the components to quantify the quality of a non-deterministic alignment. Specifically, we will provide a novel formalization of the \textit{motif-based} alignment approach introduced in~\cite{bramon2018identifying} and will uncover its geometric structure by showing that the task is analogous to a Gromov-Wasserstein \cite{Memoli2011} transport plan computation. 

\textit{Network motifs} are small subnetworks that enumerate all possible interaction patterns among species. They are widely used to capture \emph{local structural information} in ecological networks, including food webs, mutualistic systems, and multi-trophic interaction networks \cite{milo2002network, baiser2016motifs, tavella2022using, cenci2018rethinking, paulau2015motif}. In ecological applications, motif structure has been linked to interpretable functional roles (e.g., basal producers, intermediate consumers, top predators) and to properties such as stability, robustness, and energy flow. Motif-role profiles quantify structural dissimilarity in species’ ecological roles across networks by recording how often each species occupies distinct topological positions~\cite{stouffer2012evolutionary}. 

Given a species $s$ in a food web $G$, its \emph{motif-role profile} is a vector $m(s) \in \mathbb{R}^{p}$ whose $i$-th entry counts how many times $s$ appears in role $i$ across all occurrences of directed $k$-node motifs in $G$. For example, there are 13 non-isomorphic 3-node motifs, which yield $p=30$ unique roles \cite{mora2018pymfinder}. Given two species $i, j$ from possibly different webs, their dissimilarity is computed as $d\big(m(i),m(j)\big)$, for some discrepancy $d: \mathbb{R}^{p} \times \mathbb{R}^{p} \rightarrow \mathbb{R}_{\geq0}$, e.g.,  $d(x, y) := 1 - \rho(x, y)$ where $\rho(x, y)$ represents the Pearson correlation coefficient between $x$ and $y$. In the context of network alignment between $G_1$ and $G_2$, we compute the dissimilarities between every pair of species from the two networks, yielding the cost matrix $C_{ij} = d(m(i), m(j))$ for the alignment formulations.

Our objective is to find an optimal non-deterministic alignment between graphs $G_1$ and $G_2$ by solving the following optimization problem 

\vspace{-0.6cm}
\begin{align}\label{eq:optimal non-det alignment}
    \min_{T \in [0, 1]^{m \times n}}  &
    \ \alpha \langle T, A_1(C \odot T)A_2 \rangle 
    {+} (1{-}\alpha)\langle C, T \rangle 
    {-} \epsilon \langle T, {\mu} {\nu}^\top \rangle \nonumber \\ & \quad \text{s.t. }T \in \mathcal{C}_1({\mu}) \cap \mathcal{C}_2({\nu}).
\end{align}

\noindent where  $\alpha \in [0,1]$ is a tradeoff parameter between first-order and second-order alignments, $\epsilon > 0$ self-alignment penalty, and $A_1, A_2$ are the adjacency matrices of the underlying undirected graphs of $G_1, G_2$, respectively.

Similar to the Fused Gromov-Wasserstein approach introduced in \cite{vayer2020, kai2025},  $\alpha=0$ emphasizes direct role matching similar to the Wasserstein distance, while $\alpha=1$ emphasizes neighbor-role consistency similar to the Gromov-Wasserstein approach. Moreover,  $\epsilon > 0$ controls when no alignment should be preferred over a weak alignment. 

We propose a KL proximal point method with a Dykstra subroutine to enforce constraints and compute approximately optimal non-deterministic alignments as described in the following iterations:

\vspace{-0.2cm}
    \begin{align} \label{eq:update}
        \hat{T}^{(k)} &{=} P_{\mathcal{C}_1} \big(T^{(k)} \odot \exp( {-} \gamma^{-1}Q^{(k)})),  \nonumber \\
         \quad T^{(k+1)} & {=} P_{\mathcal{C}_2} \big(\hat{T}^{(k)} \odot \exp( {-} \gamma^{-1}Q^{'(k)})), \\
        Q^{(k)}& {:=} \alpha A_1(C {\odot} T^{(k)})A_2 {+} \frac{1}{2}(1 {-} \alpha)C {-} \frac{1}{2}\epsilon \mu\nu^\top,  \nonumber \\ \quad Q^{'(k)} & {:=} \alpha\, C {\odot} (A_1 \hat{T}^{(k)} A_2) {+} \frac{1}{2}(1 {-} \alpha)C {-} \frac{1}{2}\epsilon\,\mu\nu^\top. \nonumber
    \end{align}

\begin{assumption}\label{assum:aae}
    The accumulative asymmetrical error (AAE) of the iterates generated by Eq.~\eqref{eq:update} is bounded, i.e., 

    \vspace{-0.2cm}
    \begin{align*}
        \sum_{k=0}^\infty \big(D_h(\pi^{(k+1)}, w^{(k)}) - D_h(w^{(k)}, \pi^{(k+1)})\big) < \infty,
    \end{align*}
    \vspace{-0.2cm}
    
    \noindent where $D_h(x,y)$ is the Bregman divergence between $x$ and $y$ associated with a strictly convex function $h$.
\end{assumption}

Assumption~\ref{assum:aae} has been shown to hold for $h$ quadratic, and has been empirically verified for  $h$ being the relative entropy \cite[Figure 3]{li2023convergent}. 

\begin{proposition} \label{prop:convergence}
Let $G_1$, $G_2$ be two graphs with adjacency matrices $A_1, A_2$, $C$ a cost matrix, $\alpha\in[0,1]$, and $\epsilon>0$, and $T^{(0)}=\frac{1}{mn}\mathbbm{1}_m\mathbbm{1}_n^\top$. Let Assumption~\ref{assum:aae} hold. Then, the iterates generated by Eq.~\eqref{eq:update} converge to a stationary point of~\eqref{eq:optimal non-det alignment}
\end{proposition}

Deterministic alignments arise as a special case with~$T_{ij}\in\{0,1\}$, and $\mu$ and $\nu$ are uniform. Our next theoretical result formalizes the heuristic developed in~\cite{bramon2018identifying} for the alignment computation and shows its geometric structure as an optimization problem with a Gromov-Wasserstein-like cost functional.

\begin{proposition}\label{prop:det2}
Let $G_1$, $G_2$ be two graphs with adjacency matrices $A_1, A_2$ respectively, $C$ a cost matrix, then an alignment 

\vspace{-0,2cm}
  \begin{align}\label{eq:optimal det alignment}
   \hat{T}^* & = \arg\min_{T \in \{0,1\}^{m \times n}} 
    \  \langle T, A_1(C \odot T)A_2 \rangle 
       + \langle T, \Xi(T)\rangle, \nonumber \\ & \text{s.t. }T \in \mathcal{C}_1(\mathbbm{1}_m) \cap \mathcal{C}_2(\mathbbm{1}_n), \ \quad \text{where} \\
        \Xi(T) & = (1{-}\epsilon)\left(\max \left(A_1 \mathbbm{1}_m \mathbbm{1}_n^\top, \mathbbm{1}_m \mathbbm{1}_n^\top A_2\right) - A_1 T A_2\right) . \nonumber
\end{align}
    is optimal if and only if it minimizes \cite[Eq. (3)]{bramon2018identifying}.
\end{proposition}

We examine the various properties of deterministic and non-deterministic alignments in \cite[Appendix~A.2]{xu2025foodwebs}. Moreover, algorithmic details and proofs for Proposition~\ref{prop:convergence} and \ref{prop:det2} can be found in \cite[Appendices~A.1, B.1., B.2.]{xu2025foodwebs}.


\section{Non-deterministic  Backbones} 

Consider a set of networks $\mathcal{D} = \{G_1, G_2, \cdots, G_N\}$, the pairwise non-deterministic alignments between them $\{T^{ij}\}_{i,j=1}^N$ and the pairwise dissimilarity matrices between them $\{C^{ij}\}_{i,j=1}^N$. Moreover, let $[N]$ denote the set $\{1, \cdots, N\}$ for all positive integers up to $N$.

\begin{definition}[Non-deterministic role similarity score] \label{def:role-sim}
    Given a graph $G_i$ for some $i \in [N]$ and an ordering of its species, the role similarity score of species $k$ in $V_i$ is defined as

    \vspace{-6mm}
    \begin{align*}
        \sum_{p=1}^N \sum_{q=1}^{|V_p|} (1 - C^{ip}_{jq} )T_{jq}^{ip}.
    \end{align*}
    \vspace{-6mm}
\end{definition}

A species' role similarity score is high when: 1) a small amount of its mass is self-aligned, and 2)  the alignment incident to this species generates low direct-matching cost. Under this definition, a species with a high role similarity score is said to be \textbf{well-aligned} with respect to the dataset $\mathcal{D}$.

\begin{definition}[Top-k backbone of interactions]\label{def:top_k}
    Given a graph $G_i$ and the role similarity scores of its species, for an integer $k \geq 2$, we define the top-$k$ backbone of $G_i$, denoted as $B_i = (V_{B_i}, E_{B_i})$, to be the subgraph induced by the $k$ species with the highest role similarity scores.
\end{definition}
\begin{definition}[Non-deterministic transitivity score]\label{def:transitivity}
    For each species in $V_{B_i}$, which is the $j$-th species in $V_i$, its non-deterministic transitivity score is:

    \vspace{-6mm}
    \begin{align*}
    \text{Transitivity}(j) = \frac{\displaystyle \sum_{\substack{p, q \in [N] \\ p < q; p,q \neq i}}
    \displaystyle \sum_{j_p \in V_p} \sum_{j_q \in V_q} T^{ip}_{j j_p} \cdot T^{iq}_{j j_q} \cdot T^{pq}_{j_p j_q}
    }{
    \displaystyle \sum_{\substack{p, q \in [N] \\ p < q; p,q \neq i}}\sum_{j_p \in V_p} \sum_{j_q \in V_q} T^{ip}_{j j_p} \cdot T^{iq}_{j j_q}
    }.
    \end{align*}
\vspace{-4mm}
    
    The transitivity score of a food web is the average among its backbone species.
\end{definition}

Intuitively, for each species $j$ in the backbone $B_i$, we examine its alignment to species $j_p$ in $G_p$ and $j_q$ in $G_q$, across all pairs of networks $G_p$, $G_q$ not equal to $G_i$. The numerator accumulates these alignment strengths weighted by how strongly $j_p$ and $j_q$ are themselves aligned (as measured by $T^{pq}_{j_p j_q}$), representing how well the species alignments form a transitive triangle. The denominator serves as a normalization factor, ensuring the score lies between 0 and 1, by aggregating all alignment-strength products $T^{ip}_{j j_p} \cdot T^{iq}_{j j_q}$ regardless of whether $j_p$ and $j_q$ align.

\section{Numerical analysis on Sub-Saharan Africa mammal food webs}

\vspace{-0.1cm}



\begin{figure}[!t]
  \centering
  \setlength{\tabcolsep}{6pt} 
  \renewcommand{\arraystretch}{0.5}

  \newcommand{\cellw}{0.28\linewidth}

  \begin{tabular}{@{}ccccc@{}}
 
      \textbf{Ours, $\alpha=0$} &  \textbf{Ours, $\alpha=0.5$} & \textbf{Ours, $\alpha = 1$}\\ 
      \multirow{-4}{*}{\rotatebox[origin=c]{90}{$\epsilon = 0$}} \includegraphics[width=\cellw]{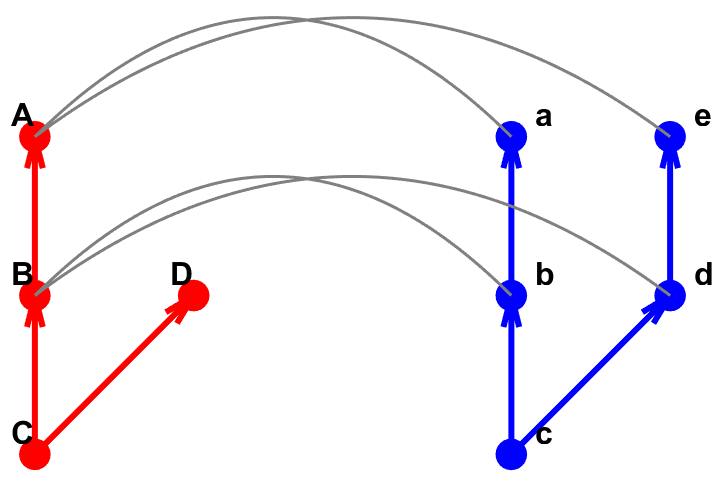}  
      & \includegraphics[width=\cellw]{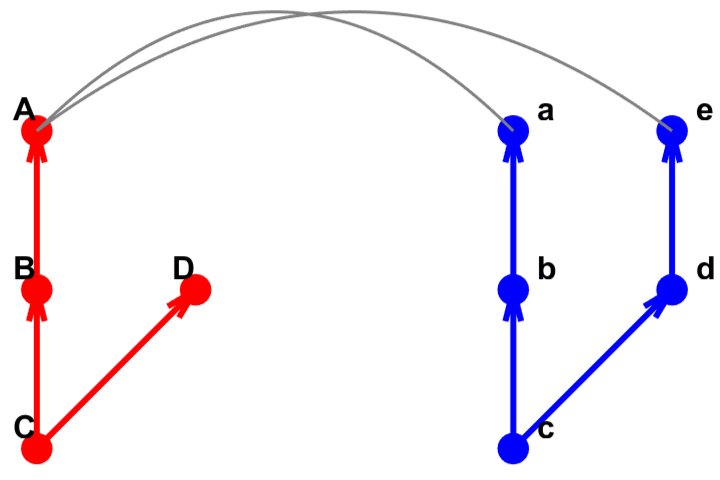} &
      \includegraphics[width=\cellw]{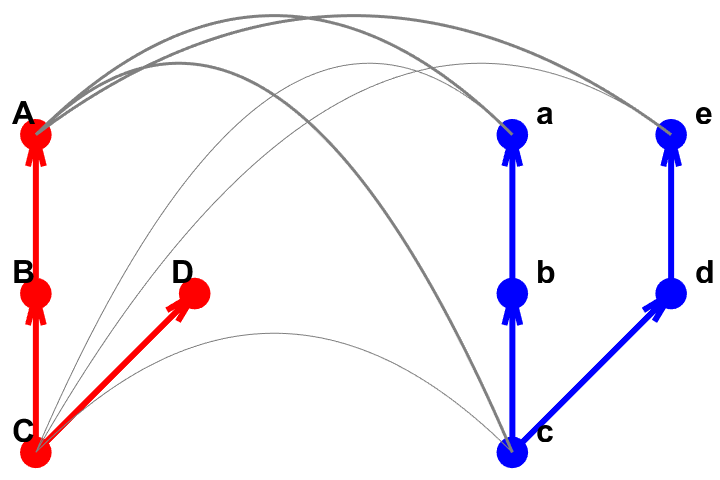} \\
 \multirow{-4}{*}{\rotatebox[origin=c]{90}{$\epsilon = 5$}}     \includegraphics[width=\cellw]{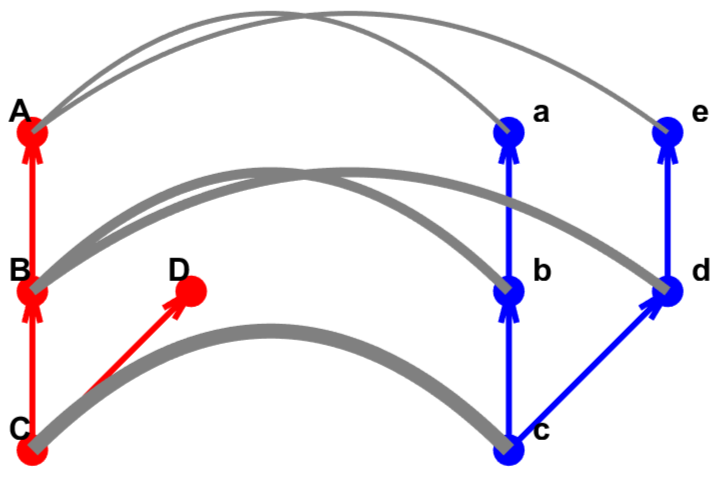} 
     & \includegraphics[width=\cellw]{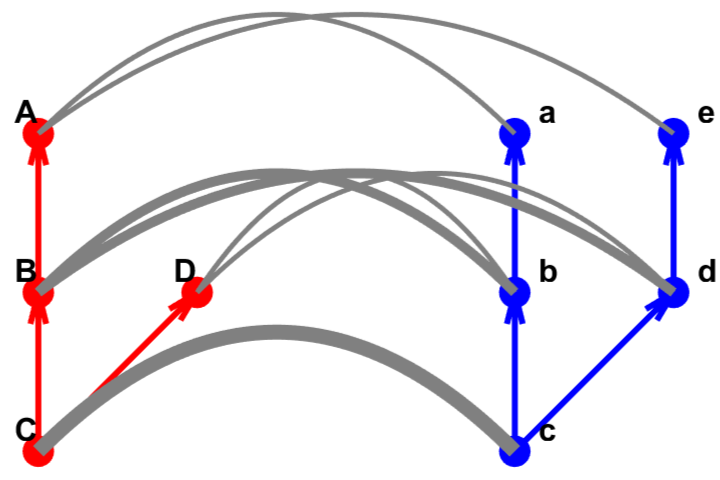} &
      \includegraphics[width=\cellw]{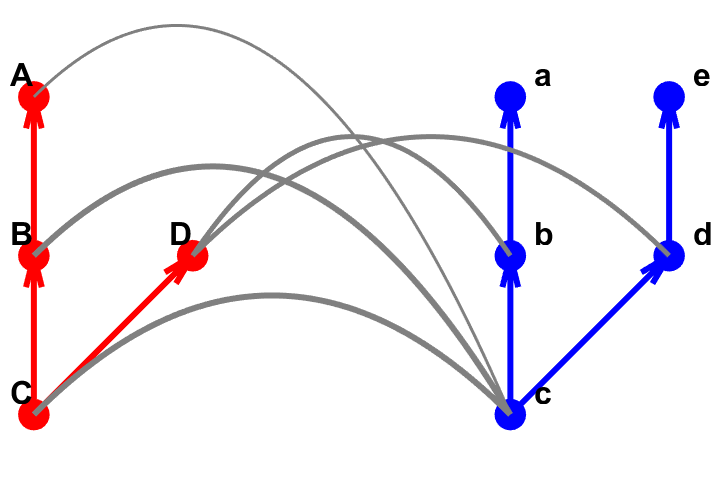} \\
  \end{tabular}

   \begin{tabular}{@{}cc@{}}
       \multicolumn{2}{c}{\textbf{Muritz~\cite{bramon2018identifying}}} \\
 \multirow{-4}{*}{\rotatebox[origin=c]{90}{$\epsilon = 0$}} \includegraphics[width=\cellw]{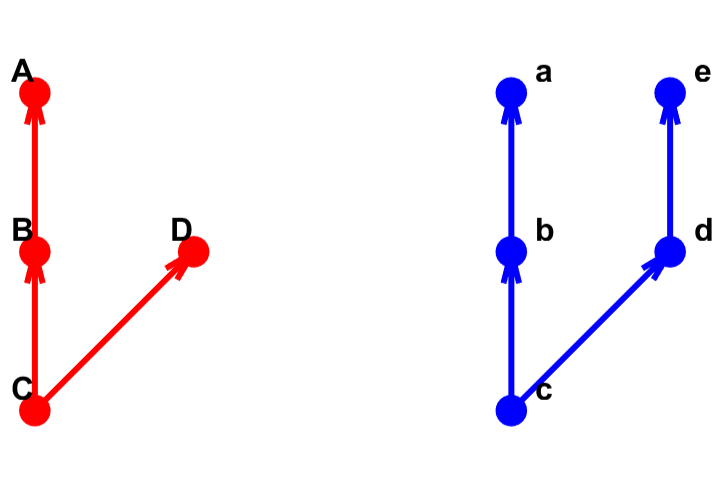} & 
\multirow{-4}{*}{\rotatebox[origin=c]{90}{$\epsilon = 5$}} \includegraphics[width=\cellw]{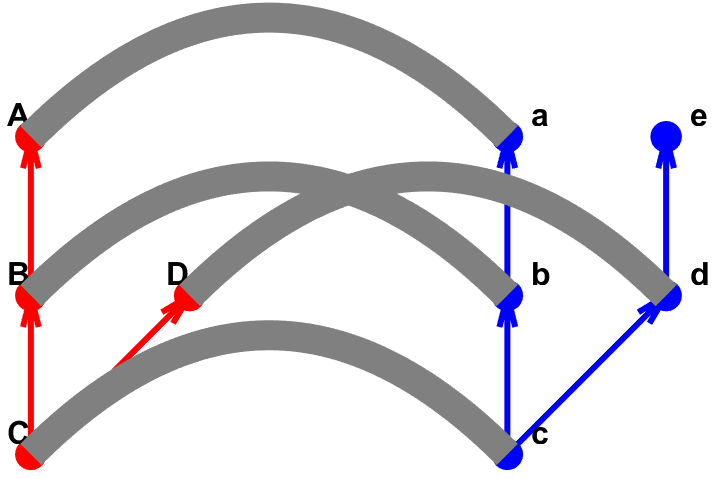}
      
  \end{tabular}
\vspace{-0.1cm}
  \caption{Optimal alignment between two toy networks with various alignment methods. Thicker lines indicate stronger alignment. Multiple deterministic alignments exist, so an arbitrary one is picked for both values of $\epsilon$.}
  \label{fig:two-by-four}
\end{figure}

We first evaluate our method on a pair of small synthetic networks, setting the species significance distributions $\mu$ and $\nu$ to be uniform. Figure \ref{fig:two-by-four} compares deterministic and non-deterministic alignments under varying values of the tradeoff parameter $\alpha$ and the self-alignment penalty $\epsilon$.
Species \textbf{A} in the red network shares an identical role profile to both species \textbf{a} and \textbf{e} in the blue network; likewise, Species \textbf{B} matches both species \textbf{b} and \textbf{d}. Deterministic alignment captures only a single pairing per species, while non-deterministic alignment identifies \textit{all} valid pairings and assigns each an alignment strength proportional to role similarity. This highlights the ability of non-deterministic alignment to detect structural redundancies.

\vspace{-0.1cm}
\begin{figure}[t]
  \centering
  \setlength{\tabcolsep}{6pt} 
  \renewcommand{\arraystretch}{0.1}

  \newcommand{\cellw}{0.7\linewidth}

  \begin{tabular}{@{}cc@{}}
      \subcaptionbox{ Deterministic alignment (Muritz~\cite{bramon2018identifying})\label{fig:hard_eatmap}}{\includegraphics[width=\cellw]{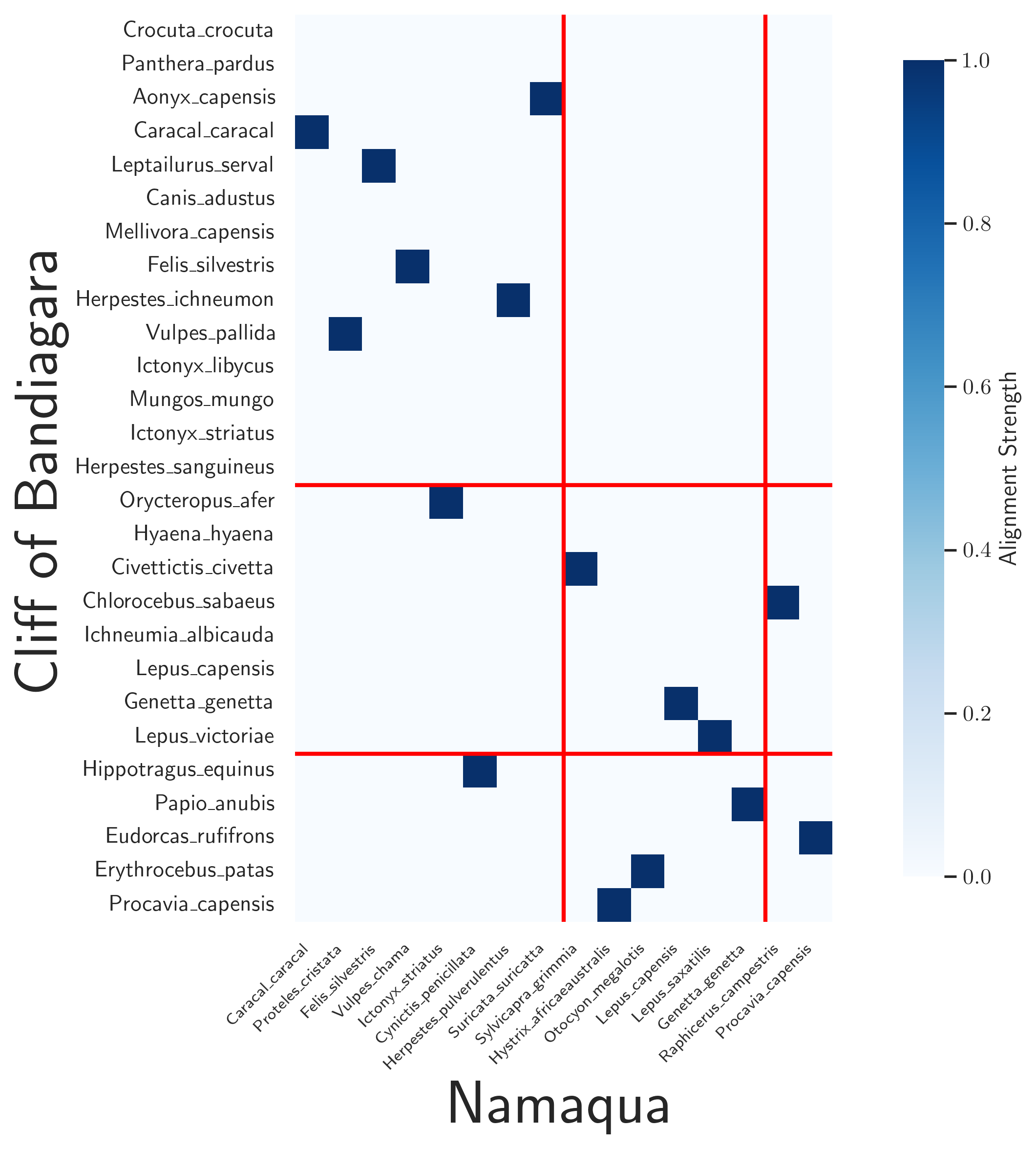}} \\
      \subcaptionbox{Non-deterministic alignment ($\alpha = 0.5$)\label{fig:soft_heatmap}}{\includegraphics[width=\cellw]{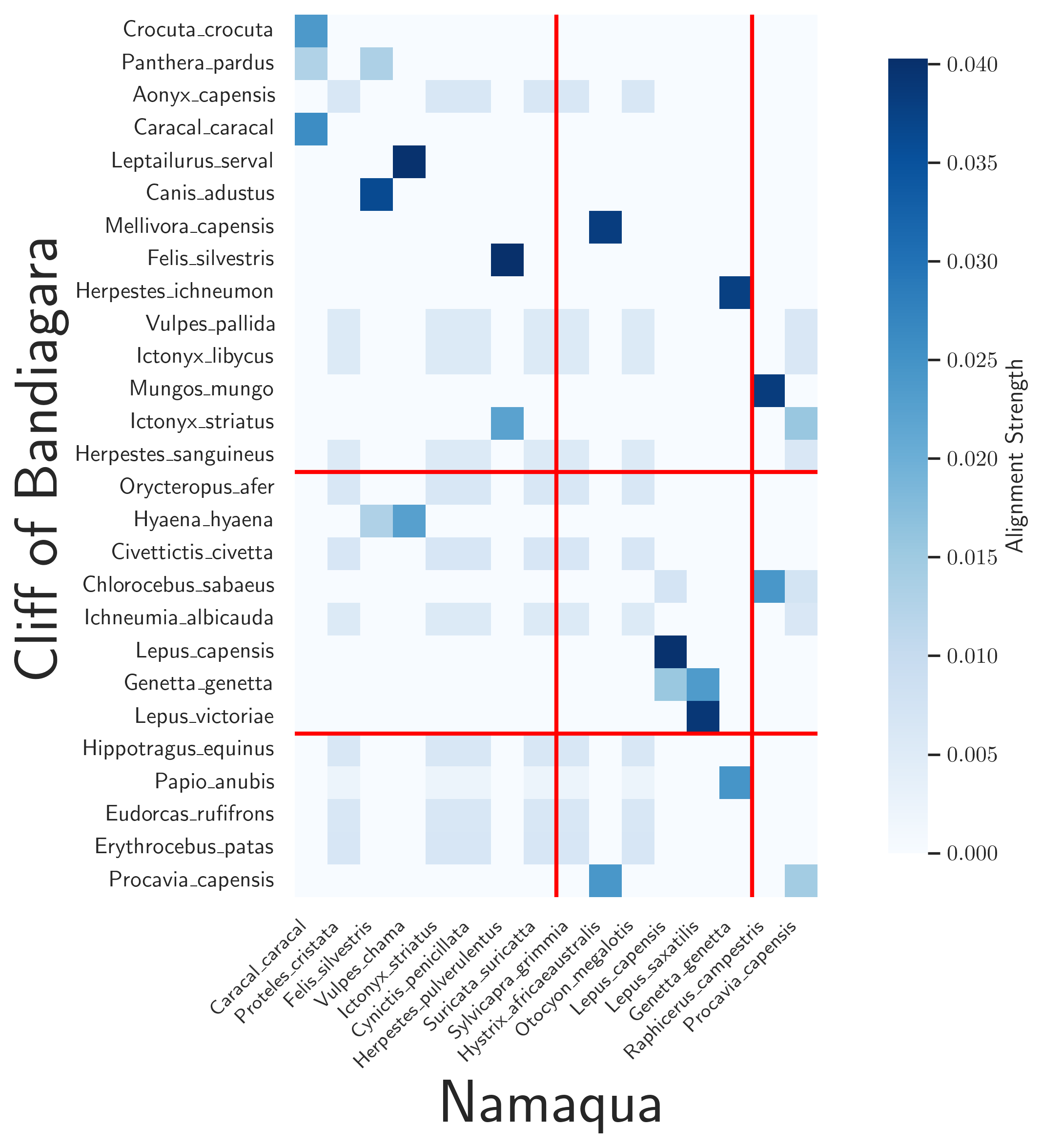}} 
  \end{tabular}
  \vspace{-0.2cm}
  \caption{Alignment heatmaps between the species in Cliff of Bandiagara and the species in Namaqua. Values in the heatmap represent the strength of alignment between species.}
  \label{fig:comparison}
  \vspace{-0.2cm}
\end{figure}

\begin{figure}[ht]
  \centering
  \begin{subfigure}[b]{0.23\textwidth}
    \centering
    \includegraphics[width=1\linewidth]{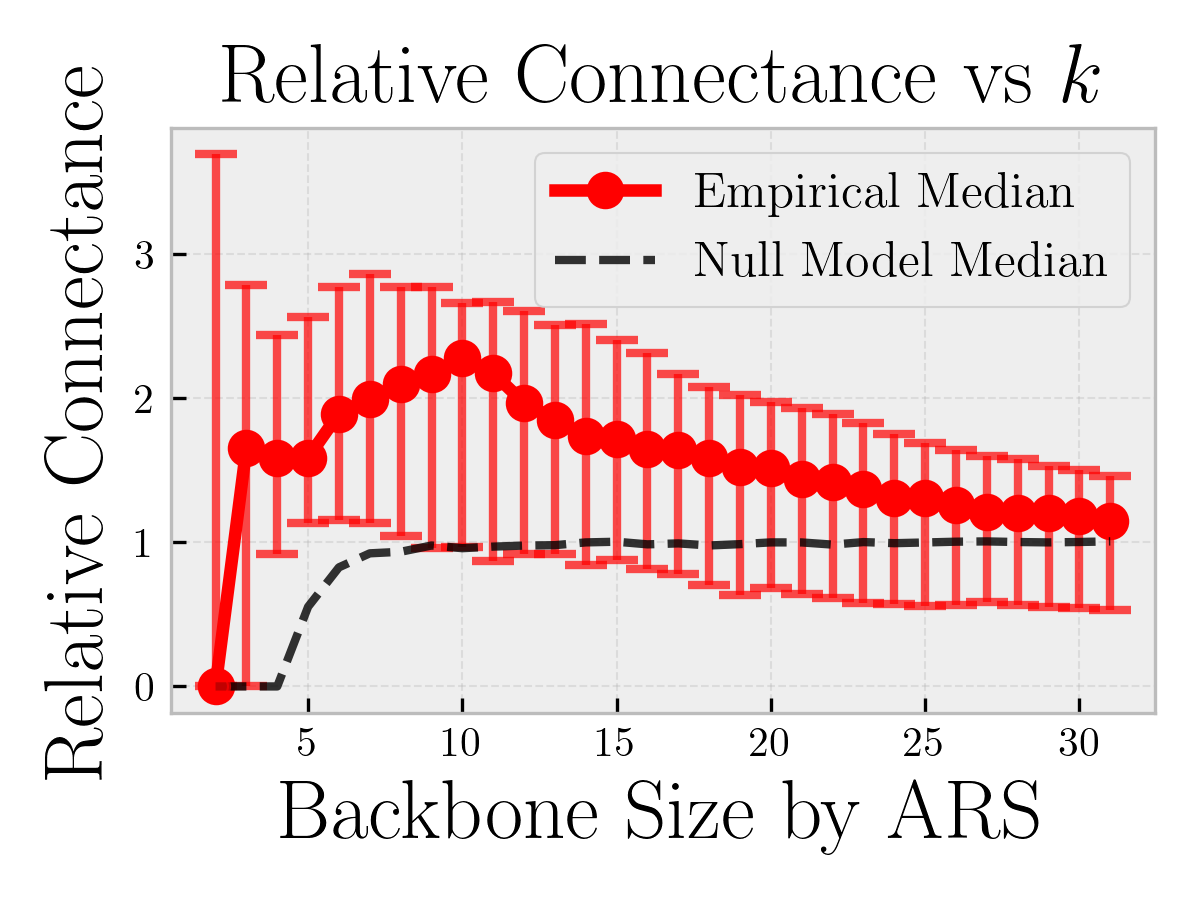}
    \caption{Relative connectance, the ratio of backbone connectance to that of the full food web.}
    \label{fig:backbone-relconn}
  \end{subfigure}\hfill
  \begin{subfigure}[b]{0.23\textwidth}
    \centering
    \includegraphics[width=1\linewidth]{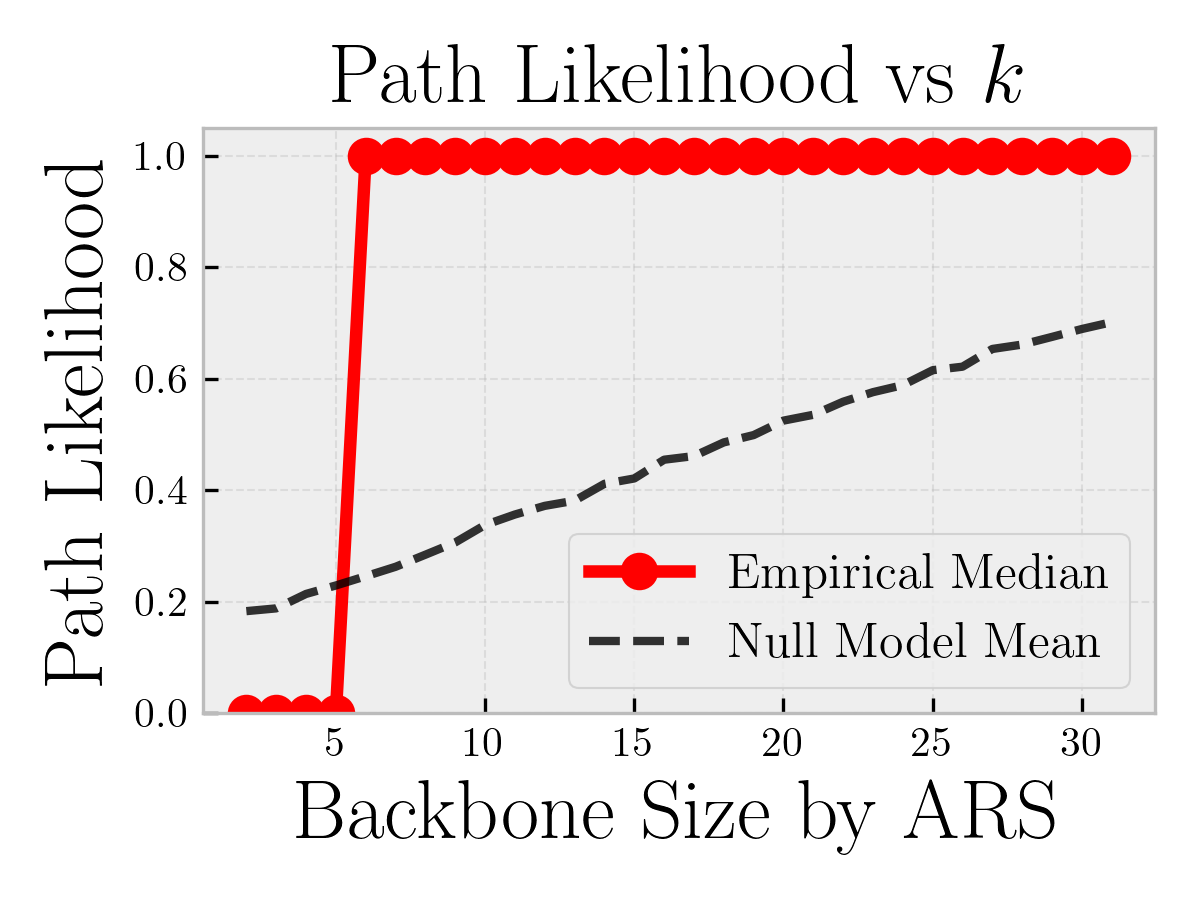}
    \caption{Path likelihood, the fraction of backbones that induce a connected subnetwork.}
    \label{fig:backbone-pathlik}
  \end{subfigure}\hfill \\
  \begin{subfigure}[b]{0.23\textwidth}
    \centering
    \includegraphics[width=0.99\linewidth]{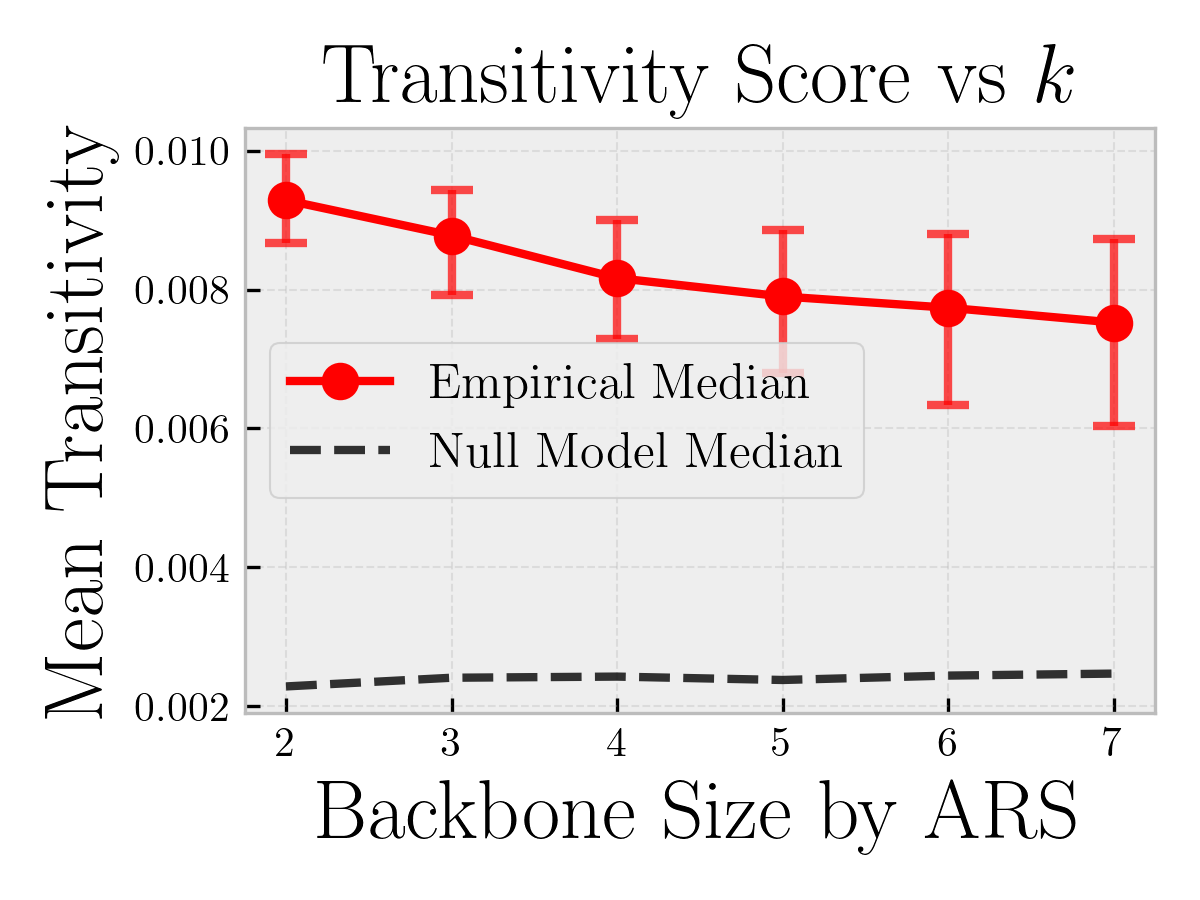}
    \caption{Transitivity, the non-deterministic alignment transitivity score.}
    \label{fig:backbone-trans}
  \end{subfigure}

  \caption{Structural properties of our top-$k$ backbones as a function of $k$.}
  \label{fig:backbone-exist}
  \vspace{-0.2cm}
\end{figure}

We next applied our method to a real-world ecological dataset, performing pairwise alignments between 129 mammal food webs across Sub-Saharan Africa, using the dataset compiled by~\cite{kai2025}, see \cite[Appendix~A.3]{xu2025foodwebs} for additional details. Non-deterministic alignments completed in $1.00\pm0.20$ seconds on average, substantially faster than deterministic alignments, which required $7.92\pm2.31$ seconds. Figure~\ref{fig:comparison} illustrates the alignment between the Cliff of Bandiagara and Namaqua networks. Both deterministic and non-deterministic approaches recover expected matches, such as \textit{Caracal caracal} in both networks and \textit{Leptailurus serval} (Bandiagara) with \textit{Felis silvestris} (Namaqua). However, the non-deterministic framework reveals additional ecological role redundancies by relaxing the one-to-one constraint. For instance, while the deterministic alignment maps \textit{Genetta genetta} (Namaqua) solely to \textit{Papio anubis} (Bandiagara), the non-deterministic alignment uncovers that \textit{Herpestes ichneumon} (Bandiagara) also occupies a similar ecological role. This example illustrates a general limitation of deterministic methods: by enforcing strict one-to-one pairings, they obscure many ecologically valid correspondences and arbitrarily discard alternatives when species roles partially overlap. Because backbones are derived from these alignments, such restrictions propagate to backbone identification and can yield incomplete or even misleading backbones. In contrast, our non-deterministic framework systematically captures all valid correspondences within a single reproducible solution, thereby producing backbones that more faithfully represent the redundancies and structural regularities present in ecological networks.

We extract the backbones of interactions, sets of interactions that persist across ecological networks, using our pairwise alignments. We followed the backbone definition in~\cite{bramon2018identifying} when applicable, and made extensions to non-deterministic alignments when necessary; see \cite[Appendix~A.4]{xu2025foodwebs} for details.

Figure \ref{fig:backbone-exist} summarizes three key properties of our top-k backbones: relative connectance, path likelihood, and transitivity. Each plot compares these properties against the corresponding averages (or medians) under a null model, in which 
$k$ species are randomly selected 50 times (for relative connectance and path likelihood) and 25 times (for transitivity score), without replacement, for each food web. Across all measures, the backbones revealed by our non-deterministic alignments consistently exhibited both higher connectance and transitivity than the null model, suggesting that their connectivity and transitivity are statistically significant. These findings show that our method yields backbones that consistently meet the backbone criteria proposed in \cite{bramon2018identifying}, while offering scalability and reproducibility.


\section{Conclusions and future work}

We introduced a non-deterministic alignment framework for identifying the backbone of interactions in ecological networks. Our method relaxes one-to-one role correspondences, allowing the capture of overlapping and redundant ecological roles. Our experiments primarily considered uniform distributions of the species importance parameters, $\mu$ and $\nu$. Future work should incorporate ecologically meaningful weights such as degree, biomass, or abundance. Parameter-free backbone constructions can be developed, such as forward backbones. Moreover, other criteria for species should be explored, such as the average entropy of their alignment distributions, which would highlight species with consistently unique roles across networks.

\newpage 

\section{Acknowledgments}
This work is funded by the National Science Foundation under Grants \#2213568 and \#2443064.

\bibliographystyle{ieeetr}
\bibliography{ref}

@misc{xu2025foodwebs,
  title        = {Identifying Common Backbones of Interactions Underlying Food Webs via Non-Deterministic Alignments},
  author       = {Yifan Xu and Carlos Taveras and Lydia Beaudrot and César A. Uribe},
  year         = {2026},
  howpublished = {Manuscript},
  note         = {Available at: \url{https://cauribe.rice.edu/xu_et_al_icassp2026_appendix/}}
}

@article{stouffer2012evolutionary,
  title={Evolutionary conservation of species’ roles in food webs},
  author={Stouffer, Daniel B and Sales-Pardo, Marta and Sirer, M Irmak and Bascompte, Jordi},
  journal={Science},
  volume={335},
  number={6075},
  pages={1489--1492},
  year={2012},
  publisher={American Association for the Advancement of Science}
}

@article{mora2018pymfinder,
  title={pymfinder: a tool for the motif analysis of binary and quantitative complex networks},
  author={Mora, Bernat Bramon and Cirtwill, Alyssa R and Stouffer, Daniel B},
  journal={BioRxiv},
  pages={364703},
  year={2018},
  publisher={Cold Spring Harbor Laboratory}
}

@article{thebault2010stability,
  title={Stability of ecological communities and the architecture of mutualistic and trophic networks},
  author={Th{\'e}bault, Elisa and Fontaine, Colin},
  journal={Science},
  volume={329},
  number={5993},
  pages={853--856},
  year={2010},
  publisher={American Association for the Advancement of Science}
}

@article{barnes2018energy,
  title={Energy flux: the link between multitrophic biodiversity and ecosystem functioning},
  author={Barnes, Andrew D and Jochum, Malte and Lefcheck, Jonathan S and Eisenhauer, Nico and Scherber, Christoph and O’Connor, Mary I and de Ruiter, Peter and Brose, Ulrich},
  journal={Trends in ecology \& evolution},
  volume={33},
  number={3},
  pages={186--197},
  year={2018},
  publisher={Elsevier}
}

@article{oliver2015biodiversity,
  title={Biodiversity and resilience of ecosystem functions},
  author={Oliver, Tom H and Heard, Matthew S and Isaac, Nick JB and Roy, David B and Procter, Deborah and Eigenbrod, Felix and Freckleton, Rob and Hector, Andy and Orme, C David L and Petchey, Owen L and others},
  journal={Trends in ecology \& evolution},
  volume={30},
  number={11},
  pages={673--684},
  year={2015},
  publisher={Elsevier}
}

@article{lurgi2012climate,
  title={Climate change impacts on body size and food web structure on mountain ecosystems},
  author={Lurgi, Miguel and Lopez, Bernat C and Montoya, Jos{\'e} M},
  journal={Philosophical Transactions of the Royal Society B: Biological Sciences},
  volume={367},
  number={1605},
  pages={3050--3057},
  year={2012},
  publisher={The Royal Society}
}

@article{dunne2002network,
  title={Network structure and biodiversity loss in food webs: robustness increases with connectance},
  author={Dunne, Jennifer A and Williams, Richard J and Martinez, Neo D},
  journal={Ecology letters},
  volume={5},
  number={4},
  pages={558--567},
  year={2002},
  publisher={Wiley Online Library}
}

@article{williams2000simple,
  title={Simple rules yield complex food webs},
  author={Williams, Richard J and Martinez, Neo D},
  journal={Nature},
  volume={404},
  number={6774},
  pages={180--183},
  year={2000},
  publisher={Nature Publishing Group UK London}
}

@article{stouffer2011role,
  title={The role of body mass in diet contiguity and food-web structure},
  author={Stouffer, Daniel B and Rezende, Enrico L and Amaral, Lu{\'\i}s A Nunes},
  journal={Journal of Animal Ecology},
  volume={80},
  number={3},
  pages={632--639},
  year={2011},
  publisher={Wiley Online Library}
}

@inproceedings{kantorovich1942translocation,
  title={On the translocation of masses},
  author={Kantorovich, Leonid V},
  booktitle={Dokl. Akad. Nauk. USSR (NS)},
  volume={37},
  pages={199--201},
  year={1942}
}

@book{villani2008optimal,
  title={Optimal transport: old and new},
  author={Villani, C{\'e}dric and others},
  volume={338},
  year={2008},
  publisher={Springer}
}

@article{tavella2022using,
  title={Using motifs in ecological networks to identify the role of plants in crop margins for multiple agriculture functions},
  author={Tavella, Julia and Windsor, Fredric M and Rother, D{\'e}bora C and Evans, Darren M and Guimaraes Jr, Paulo R and Palacios, Tania P and Lois, Marcelo and Devoto, Mariano},
  journal={Agriculture, Ecosystems \& Environment},
  volume={331},
  pages={107912},
  year={2022},
  publisher={Elsevier}
}

@article{baiser2016motifs,
  title={Motifs in the assembly of food web networks},
  author={Baiser, Benjamin and Elhesha, Rasha and Kahveci, Tamer},
  journal={Oikos},
  volume={125},
  number={4},
  pages={480--491},
  year={2016},
  publisher={Wiley Online Library}
}

@article{paulau2015motif,
  title={Motif analysis in directed ordered networks and applications to food webs},
  author={Paulau, Pavel V and Feenders, Christoph and Blasius, Bernd},
  journal={Scientific reports},
  volume={5},
  number={1},
  pages={11926},
  year={2015},
  publisher={Nature Publishing Group UK London}
}

@article{cenci2018rethinking,
  title={Rethinking the importance of the structure of ecological networks under an environment-dependent framework},
  author={Cenci, Simone and Song, Chuliang and Saavedra, Serguei},
  journal={Ecology and evolution},
  volume={8},
  number={14},
  pages={6852--6859},
  year={2018},
  publisher={Wiley Online Library}
}

@article{milo2002network,
  title={Network motifs: simple building blocks of complex networks},
  author={Milo, Ron and Shen-Orr, Shai and Itzkovitz, Shalev and Kashtan, Nadav and Chklovskii, Dmitri and Alon, Uri},
  journal={Science},
  volume={298},
  number={5594},
  pages={824--827},
  year={2002},
  publisher={American Association for the Advancement of Science}
}

@article{Memoli2011, title={Gromov-Wasserstein Distances and the Metric Approach to Object Matching}, volume={11}, rights={http://www.springer.com/tdm}, ISSN={1615-3375, 1615-3383}, DOI={10.1007/s10208-011-9093-5}, number={4}, journal={Foundations of Computational Mathematics}, author={M{\'e}moli, Facundo}, year={2011}, month=aug, pages={417-487}, language={en} }

@article{vayer2020,
AUTHOR = {Vayer, Titouan and Chapel, Laetitia and Flamary, Remi and Tavenard, Romain and Courty, Nicolas},
TITLE = {Fused Gromov-Wasserstein Distance for Structured Objects},
JOURNAL = {Algorithms},
VOLUME = {13},
YEAR = {2020},
NUMBER = {9},
ARTICLE-NUMBER = {212},
URL = {https://www.mdpi.com/1999-4893/13/9/212},
ISSN = {1999-4893},
DOI = {10.3390/a13090212}
}

@article{kai2025,
author = {Hung, Kai M. and Beaudrot, Lydia and Finneran, Ann E. and Zalles, Alex G. and Uribe, César A.},
title = {Quantifying functionally equivalent species and ecological network dissimilarity with optimal transport distances},
journal = {Methods in Ecology and Evolution},
volume = {n/a},
number = {n/a},
pages = {},
year={2025},
doi = {https://doi.org/10.1111/2041-210x.70130},
url = {https://besjournals.onlinelibrary.wiley.com/doi/abs/10.1111/2041-210x.70130},
eprint = {https://besjournals.onlinelibrary.wiley.com/doi/pdf/10.1111/2041-210x.70130}
}

@article{li2023convergent,
  title={A convergent single-loop algorithm for relaxation of Gromov-Wasserstein in graph data},
  author={Li, Jiajin and Tang, Jianheng and Kong, Lemin and Liu, Huikang and Li, Jia and So, Anthony Man-Cho and Blanchet, Jose},
  journal={arXiv preprint arXiv:2303.06595},
  year={2023}
}

@article{bramon2018identifying,
  title={Identifying a common backbone of interactions underlying food webs from different ecosystems},
  author={Bramon Mora, Bernat and Gravel, Dominique and Gilarranz, Luis J and Poisot, Timoth{\'e}e and Stouffer, Daniel B},
  journal={Nature Communications},
  volume={9},
  number={1},
  pages={2603},
  year={2018},
  publisher={Nature Publishing Group UK London}
}

\clearpage
\appendix
\section{Supplementary materials}

\subsection{The optimization algorithm for non-deterministic alignment} \label{app:algo}

\textbf{Problem setup and splitting: }Here, we use the split variables $\pi$ and $w$ to derive Algorithm~\ref{alg:full-opt}; in the main text, both are represented by the single alignment variable $T$. 

\noindent Recall the non-deterministic alignment problem (Problem \ref{eq:optimal non-det alignment}) with inequality marginals
$\mathcal{C}_1 := \{\,\pi \in [0, 1]^{m\times n} \;:\; \pi \mathbbm{1}_n \preceq \mu\,\}$ and
$\mathcal{C}_2 := \{\,w \in [0, 1]^{m\times n} \;:\; w^\top \mathbbm{1}_m \preceq \nu\,\}$.  Similar to the splitting idea presented in \cite{li2023convergent},
we rewrite cost function as follows:
\begin{align*}
    g(\pi,w) := & \alpha\langle \pi, A_1 (C \odot w) A_2\rangle
    + 
    (1-\alpha) \langle C, \frac{1}{2}(\pi + w)\rangle
    - \\
    & \epsilon \langle \frac{1}{2}(\pi + w), \mu \nu^\top\rangle
\end{align*}
so Problem \ref{eq:optimal non-det alignment} can be rewritten as:
\begin{equation}
    \min_{\pi=w} g(\pi, w) + \mathbb{I}_{\mathcal{C}_1}(\pi) + \mathbb{I}_{\mathcal{C}_2}(w)
\end{equation}
where $\mathbb{I}_\mathcal{C}$ denotes the indicator function for the constraint set $\mathcal{C}$. Then, we can penalize the equality constraint in the cost function and construct the penalized cost function:
\begin{equation}
\label{eq:penalized-cost}
\begin{aligned}
    & \min_{\pi=w} g(\pi, w) + \mathbb{I}_{\mathcal{C}_1}(\pi) + \mathbb{I}_{\mathcal{C}_2}(w) + \gamma D_h(\pi,w) \\
    & D_h(X,Y) := \sum_{ij}\Big[ x_{ij}\log \frac{x_{ij}}{y_{ij}} - x_{ij} + y_{ij}\Big]
\end{aligned}
\end{equation}
where $D_h$ is the Bregman divergence induced by the negative entropy
$h(X)=\sum_{ij}x_{ij}\log x_{ij}$, and $\gamma>0$.

\noindent \textbf{KL-Bregman Alternating Projected Gradient \cite{li2023convergent}
: }At iteration $k$, we alternate between linearizing $g$ in each block and solve a KL–proximal step subproblem:
\begin{equation} \label{eq:update-steps}
\begin{aligned}
    \pi^{k+1}
&= \arg\min_{\pi\in\mathcal{C}_1} \Big\langle Q^{(k)},\pi\Big\rangle + \gamma  D_h(\pi, w^{k}) \\ 
w^{k+1} &= \arg\min_{w\in\mathcal{C}_2}
\Big\langle Q^{'(k)}, w\Big\rangle + \gamma D_h(w, \pi^{k+1}),
\end{aligned}
\end{equation}
where the cost surrogates $Q^{(k)}$ and $Q^{'(k)}$ are defined as:
\begin{align*}
    &Q^{(k)} = \alpha A_1(C\odot w^{k})A_2 + \frac{1}{2}(1-\alpha)C - \frac{1}{2}\epsilon \mu\nu^\top \\
    &Q^{'(k)} = \alpha\, C \odot (A_1\pi^{k+1} A_2) + \frac{1}{2}(1-\alpha)C - \frac{1}{2}\epsilon\,\mu\nu^\top.
\end{align*}
Each subproblem has the same form: minimize a linear term plus a Bregman divergence to the previous iteration, under nonnegativity and a \textbf{single} set of marginal constraints.

\noindent\textbf{Dual of the KL–prox subproblem and optimizer: }We write the generic subproblem as
\begin{equation}\label{eq:subproblem}
\begin{aligned}
    &\min_{X\in\mathcal{C}'} \langle Q, X\rangle + \gamma D_h(X,Y) 
\end{aligned}
\end{equation}
where $\mathcal{C}'$ is either $\mathcal{C}_1$ or $\mathcal{C}_2$.
First, consider \eqref{eq:subproblem} with row constraints $\mathcal{C}_1=\{X \in [0, 1]^{m\times n} : X\mathbbm{1}\preceq \mu\}$. Introducing dual variables $\lambda\in\mathbb{R}^m_{\ge 0}$, the Lagrangian is
\begin{equation}
    \mathcal{L}(X,\lambda) = \big\langle Q, X\big\rangle + \gamma D_h(X,Y) + \langle \lambda, X\mathbbm{1}-\mu\rangle
\end{equation}
Minimizing w.r.t. $X$ entrywise yields, for any $x_{ij}>0$:
\begin{align*}
    0 &= \frac{\partial \mathcal{L}}{\partial x_{ij}}
    = Q_{ij} + \gamma \log\!\frac{x_{ij}}{y_{ij}} + \lambda_i \\
    \Longrightarrow \ \ 
    x_{ij} &= y_{ij}\,\exp\!\Big(-\tfrac{Q_{ij}+\lambda_i}{\gamma}\Big)
\end{align*}
If we denote $\hat{X}_{ij} = y_{ij}\exp\Big(-\tfrac{Q_{ij}}{\gamma}\Big)$ as the unconstrained optimizer, then the primal optimizer can be written as:
\begin{equation}
X^\star_{ij} = y_{ij}\exp\Big(-\tfrac{Q_{ij}}{\gamma}\Big)\exp\Big(-\tfrac{\lambda_i^\star}{\gamma}\Big) = \hat{X}_{ij}\exp\Big(-\tfrac{\lambda_i^\star}{\gamma}\Big),
\end{equation}
with complementary slackness $\lambda_i^\star\big((X^\star \mathbbm{1})_i-\mu_i\big)=0$.
An analogous derivation with column constraints $\mathcal{C}_2=\{X\ge 0, X^\top\mathbbm{1}\preceq \nu\}$
introduces $\eta\in\mathbb{R}^n_{\ge 0}$ and yields
\begin{equation}
    X^\star_{ij} = y_{ij}\exp\Big(-\tfrac{Q'_{ij}}{\gamma}\Big)\exp\Big(-\tfrac{\eta_j^\star}{\gamma}\Big),\quad \eta_j^\star\big((X^{\star\top}\!\mathbbm{1})_j-\nu_j\big)=0.
\end{equation}

\noindent \textbf{Enforcing constraints with Dykstra-style projections:}
The structure of the Lagrangian optimizer shows that enforcing $\pi \mathbbm{1} \preceq \mu$ (or $w^\top \mathbbm{1} \preceq \nu$) is equivalent to scaling each row $i$ (or column $j$) by a factor $\exp\Big(-\tfrac{\lambda_i}{\gamma}\Big)$ (or $\exp\Big(-\tfrac{\eta_i}{\gamma}\Big)$). Complementary slackness implies that:
\begin{itemize}
    \item if the current row sum $(\hat{X}\mathbbm{1})_i \leq \mu_i$(or column sum $(\hat{X}^\top\mathbbm{1})_j \leq \nu_j$), then $\lambda_i^* = 0$ (or $\eta_i^* = 0$) and we leave the row unchanged,
    \item if $(\hat{X}\mathbbm{1})_i \geq \mu_i$ (or $(\hat{X}^\top\mathbbm{1})_j \geq \nu_j$), then $\lambda_i^* > 0$ (or $\eta_i^* > 0$) and the row is uniformly scaled until it satisfies the constraint.
\end{itemize}
Therefore, the Dykstra projection steps onto the constraint sets are simply
\begin{align*}    
    P_{\mathcal{C}_1} &= \text{diag}\left(\min\left(\frac{\mu}{\hat{X}\mathbbm{1}}, \mathbbm{1}\right)\right)\hat{X} \\ \quad P_{\mathcal{C}_2} &= \hat{X}\text{diag}\left(\min\left(\frac{\nu}{\hat{X}^\top\mathbbm{1}}, \mathbbm{1}\right)\right)
\end{align*}

\noindent \textbf{Summary: }Now, we are ready to present Algorithm \ref{alg:full-opt}. In short, we
\begin{enumerate}
    \item linearize cost surrogate around one operator: $Q^{(k)} \leftarrow \alpha A_1 (C \odot T^{(k-1)})A_2 + \frac{1}{2}(1- \alpha)C - \frac{1}{2}\epsilon \nu \mu^\top$\;
    \item do a multiplicative update with step size $\gamma$: $\hat{T} \leftarrow T^{(k-1)} \odot \exp\left( - \frac{Q^{(k)}}{\gamma}\right)$
    \item project iterate back into $\mathcal{C}_1$: row\_factor $\leftarrow \min(\frac{\mu}{\hat{T}\mathbbm{1}_n}, \mathbbm{1}_m)$ and $\hat{T} \leftarrow \text{diag}(\text{row\_factor})\hat{T}$\;
    \item repeat for the other operator and the column constraints, until a stopping criteria is met.
\end{enumerate}

\begin{algorithm}
    \caption{The KL Proximal Point Algorithm with Dykstra Subroutine}\label{alg:full-opt}
    \SetKwInOut{Input}{Input}
    \SetKwInOut{Output}{Output}

    \Input{
        - the adjacency matrices for the underlying undirected graphs $A_1, A_2$ \\
        - the dissimilarity matrix $C$ \\
        - the tradeoff parameter $\alpha \in [0, 1]$ \\
        - the self-alignment penalty $\epsilon > 0$ \\
        - the step size $\gamma > 0$ \\
        - the normalized species significance distributions $\mu, \nu$ \\
        - some additional stopping criteria for main loop (step tolerance, etc.)
    }
    \Output{$T$, the non-deterministic alignment between $G_1$ and $G_2$}
    $T^{(0)} \leftarrow \frac{1}{mn}\mathbbm{1}_m\mathbbm{1}_n^\top$\;
    \ForEach{$k = 1, 2, \cdots, \text{max\_iter}$}{
        $Q^{(k)} \leftarrow \alpha A_1 (C \odot T^{(k-1)})A_2 + \frac{1}{2}(1- \alpha)C - \frac{1}{2}\epsilon \nu \mu^\top$\;
        $\hat{T} \leftarrow T^{(k-1)} \odot \exp\left( - \frac{Q^{(k)}}{\gamma}\right)$\;
        
        row\_factor $\leftarrow \min(\frac{\mu}{\hat{T}\mathbbm{1}_n}, \mathbbm{1}_m)$\;
        $\hat{T} \leftarrow \text{diag}(\text{row\_factor})\hat{T}$\;
        $Q^{(k)'} \leftarrow \alpha C \odot (A_1 \hat{T} A_2) + \frac{1}{2}(1 - \alpha)C - \frac{1}{2}\epsilon \nu\mu^\top$\;
        $\hat{T} \leftarrow \hat{T} \odot \exp(-\frac{Q^{(k)'}}{\gamma})$\;
        col\_factor $\leftarrow \min(\frac{\nu}{\hat{T}^\top\mathbbm{1}_m}, \mathbbm{1}_n)$\;
        $T^{(k)} \leftarrow \hat{T}\text{diag}(\text{col\_factor})$\;
        Check additional stopping criteria for the main loop.\;
    }
    \Return $T^{(k)}$\;
\end{algorithm}

\subsection{Proof of Proposition~\ref{prop:convergence}} \label{app:proof-conv}
We first write the split form of~\ref{eq:optimal non-det alignment} as seen in Appendix~\ref{app:algo}:
\begin{align}
\label{eq:split-form}
    \min_{\pi = w} &g(\pi, w) + \mathbb{I}_{\mathcal{C}_1}(\pi) + \mathbb{I}_{\mathcal{C}_2}(w), \\
    \text{where }&\mathcal{C}_1 = \{X \in [0, 1]^{m \times n}: X \mathbbm{1} \preceq \mu\}, \notag \\
    & \mathcal{C}_2 = \{X \in [0, 1]^{m \times n} : X^\top \mathbbm{1} \preceq \nu\}. \notag
\end{align}
Let $h(X) = \sum_{ij} x_{ij} \log x_{ij}$ be the negative entropy function and $D_h$ be its Bregman divergence. Let $\{(\pi^{(k)}, w^{(k)})\}_{k\geq0}$ be the sequence generated by alternatively solving the subproblems in~\ref{eq:update-steps}. 
\noindent We define the potential function:
\begin{equation}
    F_\gamma(\pi, w) = g(\pi, w) + \mathbb{I}_{\mathcal{C}_1}(\pi) + \mathbb{I}_{\mathcal{C}_2}(w) + \gamma D_h(\pi, w).
\end{equation}

\begin{lemma}\label{lem:1}
    The following properties hold:
    \begin{enumerate}
        \item $f(\pi, w) = \alpha \langle \pi, A_1 (C \odot w) A_2 \rangle$ is bilinear;
        \item $\mathcal{C}_1$ and $\mathcal{C}_2$ are closed, convex polytopes;
        \item the potential function $F_\gamma$ is coercive.
    \end{enumerate}
\end{lemma}
\begin{proof}[Proof of Lemma~\ref{lem:1}]
    The three statements are proved in parallel:
    \begin{enumerate}
        \item This follows immediately due to the linearity of inner products, matrix multiplication, and the Hadamard product.
        \item \textbf{Closedness: }Both inequality constraints $X \mathbbm{1} \preceq \mu$ (resp. $X^\top \mathbbm{1} \preceq \nu$) can be represented as an intersection of closed half-spaces. Since $[0, 1]^{ m\times n}$ is a product of closed intervals which is also closed, we have that the intersection of closed sets is closed; hence, $\mathcal{C}_1$ and $\mathcal{C}_2$ are closed. \\
        \textbf{Convexity: }Both the unit box $[0, 1]^{m \times n}$ and the linear inequalities are convex, and hence their intersection is convex. \\
        \textbf{Polytope: }Both $C_1$ and $C_2$ are bounded by the unit box $[0, 1]^{m \times n}$. In addition, they are both given by finitely many linear inequalities, so they are polyhedra. Since they are both bounded polyhedra, by definition, they are polytopes.
        \item Since both $C_1$ and $C_2$ are bounded, $\mathcal{C}_1 \times \mathcal{C}_2$ is bounded. Hence, whenever $\|(\pi, w)\| \rightarrow \infty$, we must have that $(\pi, w) \notin \mathcal{C}_1 \times \mathcal{C}_2$, which implies either $\pi \notin \mathcal{C}_1$ or $w \notin \mathcal{C}_2$. In the former case, we have that $\mathbb{I}_{\mathcal{C}_1}(\pi) \rightarrow \infty$, and in the latter case we have that $\mathbb{I}_{\mathcal{C}_2}(w) \rightarrow \infty$, both of which implies $F_\gamma(\pi, w) \rightarrow \infty$, as desired.
    \end{enumerate}
\end{proof}

\begin{proof}[Proof of Proposition~\ref{prop:convergence}]
    Since the accumulative asymmetrical error is bounded and the statements in Lemma~\ref{lem:1} hold, then by \cite[Theorem 3.6]{li2023convergent}, every limit point of the sequence $\{(\pi^{(k)}, w^{(k)})\}_{k \geq 0}$ generated by~\ref{eq:update-steps} belongs to the fixed point set of the Bregman Alternating Projected Gradient. Hence, Algorithm~\ref{alg:full-opt} generates a sequence $\{T^{(k)}\}_{\geq0}$ that converges to a stationary point of problem~\ref{eq:split-form}, and therefore a local minimizer of~\ref{eq:optimal non-det alignment}. 
\end{proof}

\noindent \textbf{Notes on the accumulative asymmetrical error (AAE): }The AAE is defined as:
\begin{equation*}
    \sum_{k=0}^\infty \left(D_h(\pi^{(k+1)}, w^{(k)}) - D_h(w^{(k)}, \pi^{(k+1)})\right)
\end{equation*}
Although it is difficult to theoretically show that the AAE is bounded, previous work has shown that it is often the case in practice \cite[Figure 3]{li2023convergent}. We observe that the AAE is also often bounded in practice.

\subsection{Proof of Proposition~\ref{prop:det2}} \label{app:proof-det}
Before we commence with the proof, we provide definitions for deterministic alignments and the corresponding optimal deterministic alignment problem.

\noindent \textbf{Deterministic alignment: } Deterministic alignment (as defined in \cite{bramon2018identifying}) can be formulated as a stricter special case to our alignment model, where $T \in \{0, 1\}^{m \times n}$, and $\mu = \mathbbm{1}_{m}$ and $\nu = \mathbbm{1}_{n}$:
\begin{definition}[Feasible Deterministic Alignment] \label{def:det alignment}
    Given two networks $G_1, G_2$, a \textbf{deterministic alignment} is a binary matrix $T \in \{0,1\}^{m \times n} \cap C_1(\mathbbm{1}_m) \cap C_2(\mathbbm{1}_n)$. $T_{ij} = 1$ indicates that the alignment identifies the ecological role of species $i$ and $j$ to be similar, and $T_{ij} = 0$ otherwise.
\end{definition} 

\begin{definition}[Optimal Deterministic Alignment] \label{def:optimal det alignment}
Under the same assumptions in Def.~\ref{def:non-det alignment} and with a choice of the self-penalty matrix $\Xi(T)$, an \textbf{optimal deterministic alignment} is a solution of
\begin{align}\label{eq:det alignment}
    \min_{T \in \{0,1\}^{m \times n}} 
    \ & \alpha \langle T, A_1(C \odot T)A_2 \rangle 
    + (1-\alpha)\langle C, T \rangle 
    + \langle T, \Xi(T)\rangle\\
    \text{s.t. }& \ T \in \mathcal{C}_1(\mathbbm{1}_m) \cap \mathcal{C}_2(\mathbbm{1}_n) \notag
\end{align}
\end{definition}

\begin{proof}[Proof of Proposition~\ref{prop:det2}]
Let $G_1 = (V_1,E_1)$ and $G_2 = (V_2,E_2)$ have underlying undirected adjacency matrices $A_1, A_2$, and let $C \in \mathbb{R}^{m \times n}$ be the cost matrix (with discrepancy $1 - \rho(i, j)$. In our deterministic model (Def. \ref{def:det alignment}, \ref{def:optimal det alignment}), a feasible $T \in \{0,1\}^{m \times n} \cap \mathcal{C}_1(\mathbbm{1}_m) \cap \mathcal{C}_2(\mathbbm{1}_n)$ encodes a deterministic matching between $V_1$ and $V_2$. Define the alignment set
\begin{align*}
    \lambda(T) = &\{(i, j) \in V_1 \times V_2: T_{ij} = 1\} 
    \ \cup \\ & 
    \{(i, \emptyset) \in V_1 \times \emptyset: \sum_j T_{ij} = 0\} 
    \ \cup \\ &
    \{(\emptyset, j) \in \emptyset \times V_2: \sum_i T_{ij} = 0\}
\end{align*}
Conversely, any alignment $\lambda$ as defined in \cite{bramon2018identifying} induces such a feasible $T$.

\noindent We prove that when $\alpha = 1$ and
\begin{equation} \label{eq:penalty}
        \Xi(T) = (1-\epsilon)\left(\max \left(A_1 \mathbbm{1}_m \mathbbm{1}_n^\top, \mathbbm{1}_m \mathbbm{1}_n^\top A_2\right) - A_1 T A_2\right),
    \end{equation}
the cost function in~\ref{eq:optimal det alignment} is equal to \cite[Eq. (3)]{bramon2018identifying}, so that their minimizers coincide.
First, with $\alpha = 1$, the cross-alignment terms of~\ref{eq:optimal det alignment} collapses to
\begin{equation*}
    \langle T, A_1(C \odot T)A_2 \rangle = \sum_{i,j} T_{ij} \sum_{\hat{i} \in \mathcal{N}_{G_1}(i)} \sum_{\hat{j} \in \mathcal{N}_{G_2}(j)} C_{\hat{i}\hat{j}} T_{\hat{i} \hat{j}},
\end{equation*}
where $\mathcal{N}_G(i)$ denote the set of 1-hop neighbors of $i$ in $G$. For each aligned pair of species $x = (i, j) \in \lambda(T)$, the inner double sum ranges over the set of neighbor pairings $\lambda_x := \{(\alpha, \beta) \in \lambda(T): \alpha \in \mathcal{N}_{G_1}(i), \beta \in \mathcal{N}_{G_2}(j)\}$. Thus 
\begin{align} 
\langle T, A_1(C \odot T)A_2 \rangle   
&= 
\sum_{x \in \lambda(T)}
\sum_{(\alpha, \beta) \in \lambda_x} C_{\alpha \beta} \notag
\\ &= 
\sum_{x \in \lambda(T)} \sum_{(\alpha, \beta) \in \lambda_x} 1 - \rho(\alpha, \beta) \label{eq:cross-alignment}
\end{align}
which is exactly the cross-alignment term in \cite[Eq. (3)]{bramon2018identifying}.

Next, for some alignment $x = (i, j) \in \lambda(T)$, write the ``matched neighbor count" as
\begin{equation*}
    (A_1 T A_2)_{ij} = |\{(\alpha, \beta) \in \lambda_x\}|.
\end{equation*}
Then, the penalty matrix in~\ref{eq:penalty} gives, entrywise,
\begin{equation*}
    \Xi_{ij}(T) = (1-\epsilon) \left( \max\big((A_1 \mathbbm{1}_m)_i, (A_2 \mathbbm{1}_n)_j\big) - (A_1 T A_2)_{ij}\right).
\end{equation*}
If $T_{ij} = 1$, then the term $\max\big((A_1 \mathbbm{1}_m)_i, (A_2 \mathbbm{1}_n)_j\big) - (A_1 T A_2)_{ij}$ represents ``the number of unmatched neighbors on the higher‑degree side of the alignment," i.e., $\max(k_x^\alpha, k_x^\beta)$ in \cite[Alignment algorithm section of Supplementary Methods]{bramon2018identifying}. Therefore,
\begin{equation} \label{eq:self-penalty}
    \langle T, \Xi(T)\rangle = \sum_{x \in \lambda(T)}(1 - \epsilon)\max(k_x^\alpha, k_x^\beta) = \sum_{x \in \lambda(T)} \xi_x.
\end{equation}
Combining Equations~\ref{eq:cross-alignment} and ~\ref{eq:self-penalty}, we have
\begin{align*}
    &\langle T, A_1(C \odot T)A_2 \rangle + \langle T, \Xi(T)\rangle \\
    &= \sum_{x \in \lambda(T)} \sum_{(\alpha, \beta) \in \lambda_x} 1 - \rho(\alpha, \beta) + \sum_{x \in \lambda(T)} \xi_x \\
    &= \sum_{x \in \lambda(T)} \left( \sum_{(\alpha, \beta) \in \lambda_x} (1 - \rho(\alpha, \beta)) +  \xi_x\right),
\end{align*}
as desired.
\end{proof}

\end{document}